\documentclass[11pt]{article}           
\usepackage{epsfig,latexsym,theorem}
\usepackage{algorithm,algorithmic}

\newcommand{\lf}{\left(}
\newcommand{\rg}{\right)}
\newcommand{\ceil}[1]{\lceil #1 \rceil}
\newcommand{\Ceil}[1]{\left\lceil #1 \right\rceil}
\newcommand{\mn}[1]{\min{\{ #1 \}}}
\newcommand{\rsort}{{\tt  SqRan} }
\newcommand{\dsort}{{\tt  SqDet} }
\newcommand{\mrsort}{{\tt McRan} }
\newcommand{\mdsort}{{\tt McDet} }
\newcommand{\sort}{{\tt  SomeSort} }

\newenvironment{proof}{\noindent {\bf Proof:}}{\ \BBox \\*}
\newcommand{\BBox}{\rule{0.1in}{0.1in}}

\newtheorem{lem}{Lemma}
\newtheorem{theo}{Theorem}
\newtheorem{cla}{Claim}

\title{Using parallelism techniques to improve sequential and multi-core 
sorting performance
      }

\author{Alexandros V. Gerbessiotis \\
        CS Department \\
        New Jersey Institute of Technology \\
        Newark, NJ 07102. \\
        alexg@njit.edu \\
        Tel: (973)-596-3244 
}
\date{{\sc \today}}

\begin{document}
\maketitle

\begin{abstract}
We propose new sequential sorting operations by adapting techniques
and methods used for designing parallel sorting algorithms. 
Although the norm is to  parallelize a sequential algorithm to
improve performance, we adapt a contrarian approach: we employ
parallel computing techniques to speed up sequential sorting.
Our methods can also work for multi-core sorting with minor
adjustments that do not necessarily require full 
parallelization of the original sequential algorithm.  
The proposed approach leads to the development 
of asymptotically efficient deterministic and randomized sorting 
operations whose practical sequential and multi-core performance, 
as witnessed by an experimental study,  matches or surpasses existing 
optimized sorting algorithm implementations.

We utilize parallel sorting techniques  such as deterministic 
regular sampling and random oversampling. We extend the notion 
of deterministic regular sampling into deterministic regular 
oversampling for sequential and multi-core sorting and 
demonstrate its potential.
We then show how these techniques can be used for 
sequential sorting and also lead to better multi-core sorting 
algorithm performance as witnessed by the undertaken
experimental study.
\end{abstract}

\vskip 0.1in
\begin{center}
\end{center}


\section{Introduction}
\label{intro}

Comparison-based sequential (or serial) sorting is a computational 
problem for which software libraries provide  generic optimized 
implementations that are stable and/or can sort in-place. 
For the remainder, stable sorting is to mean that equally 
valued keys retain in the output their relative input order, and  
sorting is in-place if the same space is used to store the 
input and the output and all other extra space used is constant 
relative to the number of keys to be sorted.
Examples of such software library-based sorting implementations 
include {\tt qsort} available through the C Standard Library, 
or the {\tt sort} and {\tt stable\_sort} variants available
in the C++ Standard Template Library.

The problem of sorting keys in parallel has also been studied 
extensively. One major requirement in parallel algorithm 
design is the minimization of interprocessor communication
that reduces non-computational overhead and thus speeds up  
parallel running time. In order to shorten communication
several techniques are employed that allow coarse 
(eg. several consecutive keys  at a time) rather than fine-grained 
(eg. few individual keys) communication. 
Coarse-grained communication usually 
takes advantage of locality of reference as well.
In parallel sorting, if one wants to sort $n$ keys 
in parallel using $p$ processors, one obvious way to achieve 
this is to somehow split the $n$ keys into $p$ sequences of 
approximately the same size   and then sort these $p$ sequences 
independently and in parallel using a fast  
sequential algorithm so that the concatenation 
(as opposed to merging) of the $p$ 
sorted sequences would generate the sorted output.

For $p=2$ this essentially becomes  quicksort \cite{Hoare62}. 
Generating $p$ sequences of approximately the same size   is 
not straightforward, whether $p=2$ or $p>2$.
One way to effect a $p$-way splitting as used by \cite{Frazer70} for
external memory sequential sorting is to pick a 
{\em random sample} of  $p-1$ keys and use  them as splitters (i.e.
pivot keys).  This is easy to implement but not very effective 
in making sure that the $p$ resulting sequences effected by 
the $p-1$ splitters are balanced in size \cite{Frazer70}.
Key sampling employed in parallel sorting has been studied
in \cite{HC,Reif87,Reischuk85} that provide algorithms with 
satisfactory and scalable theoretical performance.
Using the technique of {\em random oversampling}, fully developed 
and refined in the context of parallel sorting \cite{HC,Reischuk85,Reif87}, 
one can pick the $p-1$ splitters by using a sample of 
$ps-1$ keys where $s$ is the {\em random oversampling factor}. 
After sorting the sample of size $ps-1$, one can then identify the 
$p-1$ splitters as equidistant keys in the sorted sample. 

One can trace the technique of random oversampling to  original
quicksort \cite{Hoare62}.
Quicksort allows for a variety of choices
for the splitter (pivot) key: first, last, or middle.  
Other choices can be the median of those three keys, 
or in general, the median of $2t+1$ sample keys \cite{Hoare62}, 
for some choice of $t$.
The median of three or $2t+1$ sample keys is an early instance of 
the use of the technique of  oversampling in splitter selection: 
draw a larger sample from the input to pick a better 
(or more reliable) candidate from it as the splitter.

In  \cite{Reif87} it is  shown that the $p$ subsequences of the
$n$ input keys induced by the $p-1$ splitters resulting from random
oversampling will retain with high probability $O(n/p)$ keys each 
and will thus be balanced in size up to multiplicative constant 
factors.
In \cite{Gerbessiotis94}, it is  shown that by fine-tuning and 
carefully regulating the oversampling factor $s$, 
the $p$ sequences will retain with high probability 
$(1+ \epsilon ) n/p$ keys and will thus be more finely balanced. 
Parameter  $0< \epsilon <1$ is a parameter that is controlled 
by $s$.
The bounds on processor imbalance during sorting are tighter than
those of any other random sampling or oversampling algorithm  
\cite{HC,Frazer70,Reif87,Reischuk85}.

Thus if one was to extend say a traditional quick-sort method to use 
oversampling, one would need to add more steps or phases to the 
sorting operation.
These would include a sample selection phase, a sample-sort phase, 
a splitter selection phase, followed by a phase that
splits the input keys around the chosen splitters. 
Then the split keys are to be sorted either recursively by this same 
extended method or by the traditional quick-sort method.

In this work we propose a sample-based randomized sorting operation 
\rsort that does not follow this  traditional pattern of 
sample and splitter selection but instead follows the pattern of 
deterministic regular sampling \cite{SH} and  in particular 
deterministic regular oversampling \cite{GS96}  
(to be described shortly), that are  developed in the context of 
parallel sorting \cite{SH}. 
In such an approach key sorting precedes sample and splitter 
selection similarly to multi-way merging \cite{Knuth73}, 
followed at the end by a multi-way merging that is 
less locality-sensitive.


{\em Deterministic regular sampling} \cite{SH} works as follows.
First split regularly and evenly $n$ input keys into $p$ 
sequences of equal size   (the first $n/p$ or so keys  become 
part of the first sequence, the next $n/p$ keys 
part of the second and so on).
Sort the $p$  sequences independently, and then pick from each 
sequence $p-1$ sample (and equidistant) keys for a total sample 
of size $p(p-1)$. 
A sorting (or multi-way merging) of these sample keys can be 
followed by the selection of $p-1$ equidistant splitters 
from the sorted sample. 
One can then prove that if the $n$ keys are 
split  around  the $p-1$ splitters,
then none of the $p$ resulting  sequences will be of size   more 
than $2n/p$ \cite{SH}.
In \cite{GS96,GS99a,GS97a} the notion of regular sampling is 
extended to include oversampling thus giving rise to 
{\em deterministic regular oversampling}.
In that context choosing a larger sample $p(p-1)s$ deterministically, 
with $s$ being now the {\em regular oversampling factor}, one can 
claim  that each one of the $p$ sequences split by the $p-1$ 
splitters is of size   no more than 
$(1+\delta ) n/p$, where $\delta >0$ also depends on the choice of $s$.

Randomized oversampling-based sorting is supposed to be superior to
deterministic regular sampling or oversampling-based sorting as
the regular oversampling factor $s$ can not be finely-tuned as much
as the  random oversampling factor $s$ \cite{GS96,GS99a}.




In this work we introduce a template  for sequential sorting using as
basis the deterministic sorting algorithm  introduced in 
\cite{GS96,GS99a,GS97a}.
We ``deparallelize'' a parallel
deterministic sorting algorithm that uses regular oversampling
by converting it into a sequential algorithm. Local (to a processor)
simultaneous sorting  can be done by a single processor in turn
using any fast and available sorting implementation which can include 
{\tt qsort, sort, stable\_sort} and we shall call it \sort generically.
The deterministic sorting operation that results 
will be called \dsort depicted in Algorithm~\ref{DSORT}.
If instead of regular oversampling a  random oversampling is  employed, 
but otherwise the steps remained the same,
the same template  could be used for a
randomized sorting operation to be called similarly \rsort and depicted
in Algorithm~\ref{RSORT}.
Therefore the same template  can be be used for \dsort
and \rsort except that in the former case  deterministic regular
oversampling is used \cite{GS96,GS99a,GS97a}, and
in the latter case random oversampling is used in a way that deviates from
the traditional approach of \cite{HC,Reif87,Gerbessiotis94}, where sampling
precedes key sorting.

Since both operations \dsort and \rsort are sequential, 
the choice of $p$ is not controlled by the number of available processors or cores. 
The choice of $p$ will primarily be affected by other characteristics of the 
host architecture such as multiple memory hierarchies (eg. cache memory)
that affect locality of reference. The existence of multiple cores 
can also affect the choice of $p$. 
In our discussion to follow for the
case of multi-core sorting we shall introduce parameter $m$ to be the number 
of available cores. In general, we shall assume that $p \geq m$.

Even though \dsort and \rsort  look similar, random
oversampling is provably theoretically better than 
deterministic regular oversampling. The oversampling parameter 
in \rsort can vary more widely than in \dsort
thus resulting in more balanced work-load during the multi-way 
merging phase that takes into more advantage locality of reference 
issues (eg. cache memory).

In the following section we first introduce \dsort and analyze its performance 
characteristics and then show how one can slightly modify the
sampling phase of it to generate operation  \rsort.
Then we present the multi-core variants \mdsort and \mrsort.
Finally we present some experimental results that are  derived by
implementing the proposed operations  \dsort, \rsort and \mrsort in ANSI C.
The conclusion of the experimental study  is that \dsort, \rsort, \mrsort
coupled with \sort are better than only using \sort.

\section{The \dsort and \rsort sorting operations}

We describe operations \dsort and \rsort that utilize for sorting \sort. 
Then we proceed to modifying those two operations
in a very simple manner to work for multi-core processors. The resulting
operations are \mdsort and \mrsort respectively.

\subsection{Operation \dsort}

The proposed operation \dsort is depicted in Algorithm~\ref{DSORT}
and is based on a non-iterative variant of the 
bulk-synchronous \cite{Valiant90,Valiant90b}
parallel sorting algorithm of \cite{GS96,GS99a}. It is
deterministic regular sampling based \cite{SH} but also extends 
regular sampling to deterministic regular oversampling and thus utilizes 
an efficient partitioning scheme that splits -- almost evenly and 
independently of the input distribution -- an arbitrary number of 
sorted sequences and deals with them independently of each other. 
In Section \ref{duplicate} this baseline template operation  is
augmented to handle transparently and in optimal asymptotic 
efficiency duplicate keys. In our approach duplicate handling 
does not require doubling of computation time, excessive increase
of space requirements (or in the context of parallel
computing, communication time as well)
that other regular sampling/oversampling approaches, sequential
or parallel, seem to require \cite{jaja1,jaja2,jaja3}.

%
%
\begin{algorithm}
\caption{{\dsort $(X,n,p,$ \sort $)$  \{sort $n$ keys of $X$; utilize \sort for sorting \} }}
\label{DSORT}
\begin{algorithmic}[1]
\STATE {\sc BaselineSorting.} The $n$ input keys are regularly and evenly split into $p$ 
sequences each one of size   approximately $n/p$. Each sequence
is then  sorted by \sort. Let $X_k$, $0 \leq k \leq p-1$, be the $k$-th sequence after
sorting.
\STATE {\sc Partitioning: Sample Selection.}  Let $r =  \ceil{\omega_n} $ and $s=rp$. Form 
locally a sample $T_k$ from the
sorted $X_k$.  The sample consists of $rp-1$ evenly spaced
keys of $X_k$ that partition it into $rp$ evenly sized segments; append the maximum of the
sorted $X_k$ (i.e. the last key) into $T_k$ so that the latter gets $rp$ keys. 
Merge all $T_k$ into a sorted sequence $T$ of $rp^2$ keys.
\STATE {\sc Partitioning: Splitter Selection.}  
Form the splitter sequence $S$ that contains the $(i\cdot s)$-th smallest keys of $T$,
$1 \leq i \leq p-1$, where $s=rp$. 
\STATE {\sc Partitioning: Split input keys.}  
Then split the sorted $X_k$ around $S$ into sorted subsequences $X_{k,j}$, $0 \leq j \leq p-1$, for all
$0 \leq k \leq p-1$.
\STATE  {\sc Merging.} All  sorted subsequences $X_{k,j}$ for all $0 \leq k \leq p-1$ are merged
into $Y_j$, for all $0 \leq j \leq p-1$. The concatenation of $Y_j$ for all $j$ is $Y$. 
Return $Y$.
\end{algorithmic}
\end{algorithm}

Algorithm~\ref{DSORT} describes operation \dsort($X,n,p$,  \sort).
$X$ denotes the input sequence and $n$ the number of keys of $X$.
Sorting in \dsort is performed by the supplied \sort function.
Thus \dsort serves as a template
for performing sorting that can utilize
highly-optimized existing (software library) sorting algorithms. 
The parameter $p$
is user-specified and denotes the number of  sequences input $X$
will be split into;
naturally, a wrapper function of \dsort
can set the value of $p$ to some  specific value depending on the
available processor architecture.
One can thus vary $p$ in ways that will take full advantage of
multiple memory hierarchies or the underlying processor architecture whose
behavior and characteristics might not be easily understood by the 
average programmer/user. Within \dsort parameter $s$ 
is the regular oversampling factor whose value is regulated through the
choice of $\omega_n$ that also affects $r$. 
Function $\omega_n$ could have been included in the parameter list
of \dsort. However, because
it is a function of $n$ 
its value can be set permanently through $n$ inside \dsort
after     some initial calibration/benchmarking.
The theorem  and proof that follow simplify the results shown in 
a more general context in  \cite{GS96,GS99a}.
Operation \dsort in Algorithm~\ref{DSORT} corresponds to
the simplest case of the deterministic algorithm in \cite{GS99a} 
where an analysis for all possible values of processor size versus problem size
is presented.

\begin{theo}
\label{detsorting}
For any $n$ and $p \leq n$, and any function
$\omega_{n}$ of $n$ such that $\omega_{n}=\Omega(1)$,
$\omega_n = O(\lg{n})$ and $p^2 \omega^2_{n} = o(n)$,
operation \dsort
requires time $A n \lg{(n/p)} + B n \lg{p}+ o(n \lg{n})$, 
if \sort requires $A n \lg{n}$ time to sort $n$ keys and 
it takes $B n \lg{p}$ time to merge $n$ keys of $p$ sorted sequences, 
for some constants $A,B \geq 1$.
\end{theo}

\begin{proof}{}
The input sequence is split arbitrarily into $p$ sequences of about
the same size   (plus or minus one key). This is step 1
of \dsort.
Moreover, the keys are distinct since in an extreme case,
we can always make them so by, for example, appending to them the code
for their memory location. We later explain how we handle duplicate keys
without doubling (in the worst case) the number of comparisons performed.
Parameter $r$ determines the desired upper bound in 
key imbalance of the $p$ sorted sequences $Y_k$
that will form the output.
The term $1+1 /r = 1+ 1/ \ceil{\omega_n}$ that will characterize such an imbalance
is also referred to as {\em bucket expansion}
in sampling based randomized sorting algorithms \cite{Blelloch91}.

{\em Note.} In the discussion to follow we track constant values for 
key sorting and multi-way merging but use asymptotic notation for
other low-order term operations.

In step 1, each one of the $p$ sequences is sorted independently of
each other using \sort.
As each such sequence is of size   at most  $\ceil{n/p}$, this step
requires  time $A \ceil{n/p} \lg{\ceil{n/p}}$ per sequence or a total of
$A p \ceil{n/p} \lg{\ceil{n/p}}$.
Algorithm \sort is any sequential sorting algorithm of such
performance. The overall cost of this step is $A n \lg{(n/p)} +O(p \lg{n})$.

Subsequently, within each sorted subsequence $X_k$,  
$\ceil{\omega_{n}} p - 1=rp-1$ evenly spaced sample keys are selected, 
that partition the corresponding
sequence into $r p$ evenly sized segments.  Additionally, 
the largest key of each sequence is appended to $T_k$.
Let $s = rp$ be the size   of the 
so identified sequence $T_k$.
Step 2 requires time $O(ps)=O(p^2 r)$ to perform if
the time $O(s)$ of forming one sequence is multiplied by the total number $p$
of such sample sequences.
The $p$ sorted sample sequences, each consisting of $s$ sample keys, 
are then merged or sorted into $T$. 
Let sequence $T= \langle t_{1}, t_{2}, \ldots, t_{ps}\rangle$ be 
the result of that  operation.  
The cost of step 2 can be  that of $p$-way merging i.e. 
$B ps\lg{p} =B p^2 r \lg{p}$.
In step 3,  a sequence $S$ of evenly spaced splitters is formed 
from the sorted sample by picking as splitters keys 
$t_{is}$, $1 \leq i  < p$. This step takes time $O(p)$.

Step 4 splits $X_k$ around the sample keys in $S$.
Each one of the $p$ sorted sequences  decides the position of every key 
it holds with respect to the $p-1$ splitters by
way of sequential merging the $p-1$ splitters of $S$ with the input
keys of $X_k$ in $p-1+n/p$ time per sequence. Alternately this
can be achieved by 
performing a binary search of the splitters into the sorted keys 
in time $p \lg{(n/p)}$, and subsequently counting the number of keys 
that fall into each one of the $p$ so identified subsequences induced 
by the $p-1$ splitters.
The overall running time of this step over all $p$ sequences is
thus $O(p^2 +n)$ if merging is performed or $p^2 \lg{(n/p)}$ if
binary search is performed.

In step 4,  $X_{k,j}$ is the $j$-th sorted subsequence of  
$X_k$ induced by $S$. This subsequence will become part of the 
$Y_j$-th output sequence in step 5.
In  step 5, $p$ output sequences $Y_j$ are formed that will eventually
be concatenated. Each such output sequence $Y_j$ is formed from the
at most $p$ sorted subsequences $X_{k,j}$ for all $k$, formed in step 4.
When this step is executed, by way of Lemma~\ref{balance} to be shown 
next, each $Y_j$ will comprise of at most 
$p=\mn{p, n_{\mathit{max}}}$ sorted subsequences  $X_{k,j}$ for a 
total of at most $n_{\mathit{max}}$ keys for $Y_j$, and $n$ keys for $Y$,
where
$n_{\mathit{max}} =(1+1/\ceil{\omega_{n}})  (n/p) +\ceil{\omega_n} p$.
The cost of this step is that of multi-way merging
$n$ keys by some deterministic algorithm \cite{Knuth73},
which is $B n \lg{p} $, as long as  $\omega_n^2 p = O(n/p)$, as
needed by Lemma~\ref{balance} to follow.

{\sc BaselineSorting} and {\sc Merging} thus contribute 
$A n \lg{(n/p)} + O(p \lg{n} )$ and $Bn \lg{p}$ respectively.
Sample selection and sample-sorting contributions amount to $O(p^2 r \lg{p})$.
Step 4 contributions are $O(n+p^2 )$ or $O(p^2 \lg{(n/p)})$.
Summing up all these computation terms  we get that the total runtime
is $A n \lg{(n/p)} + B n \lg{p}+ O(p^2 r \lg{n} + n)$. 
If $p^2 r = o(n)$, 
this is $A n \lg{n} +(A-B) n \lg{p} + o(n \lg{n} )$. 
Note that in the statement of the theorem we
use a stronger condition $p^2 r^2 = p^2 \omega_n^2 = o(n)$ for
Lemma~\ref{balance} to be applicable.
\end{proof}

It remains to show  that at the completion of  step 4 the input keys are
partitioned into (almost) evenly sized subsequences.
The main result is summarized in the following lemma.


\begin{lem}
\label{balance}
The maximum number of keys $n_{\mathit{max}}$ per output sequence  $Y_j$
in \dsort is given by\\
  $(1+1/\ceil{\omega_{n}})  (n/p) +\ceil{\omega_n} p$, for any
$\omega_{n}$ such that
$\omega_{n} = \Omega(1)$ and $\omega_n = O(\lg{n})$, provided that
$\omega_n^2 p = O(n/p)$ is also satisfied.
\end{lem}

\begin{proof}{}
Although it is not explicitly mentioned in the description of
algorithm \dsort we may assume that we initially pad the
input so that each sequence has  exactly $\ceil{n/p}$ keys. At most
one key is added to each sequence (the maximum key can be such a
choice).  Before performing the sample selection operation, we also pad
the input so that afterwards, all segments have the same number of keys
that is, $x =\ceil{\ceil{n/p}/s}$. The padding operation requires time
at most $O(s)$, which is within the lower order terms
of the analysis of Theorem~\ref{detsorting}, and therefore, does not
affect the asymptotic complexity of the algorithm.  We note that padding
operations introduce duplicate keys; a discussion of duplicate handling
follows this proof.

Consider an arbitrary splitter $t_{is}$, where $1 \leq i < p$.
There are at least $i s x$ keys which are not larger than $s_{is}$,
since there are $i s$ segments
each of size   $x$ whose keys are not larger than $s_{is}$.  Likewise,
there are at least $(ps - i s - p + 1) x $ keys which are not smaller
than $s_{is}$, since there are $p s - i s - p + 1$ segments each
of size   $x$ whose keys are not smaller than $s_{is}$.  Thus, by noting
that the total number of keys has been increased (by way of padding
operations) from $n$ to $p s x$, the number of keys
$b_{i}$ that are smaller than $s_{is}$ is bounded as follows.
\[
 i s x \leq b_{i} \leq p s x - \lf( ps- i s - p + 1 \rg)  x.
\]
A similar bound can be obtained for $b_{i+1}$.  Substituting
$s= \ceil{\omega_{n}} p$ we therefore
conclude the following.
\[
  b_{i+1} - b_{i  } \leq  sx +px-x \leq sx+px =
    \Ceil{\omega_{n}} p x + p x.
\]
The difference $n_{i} = b_{i+1} - b_{ i }$ is independent
of $i$ and gives the maximum number of keys per split sequence.
Considering that $x \leq (n + p s)/(ps)$ and substituting $s=
\ceil{\omega_{n}} p $, the following bound is
derived.
\[
  n_{\mathit{max}} =
  \lf 1+\frac{1}{\Ceil{\omega_{n}} }\rg
      \frac{n+ps}{p}.
\]
\noindent
By substituting in the  numerator of the previous expression
$s= \ceil{\omega_{n}} p$,
we conclude that the maximum number of keys $n_{\mathit{max}}$ per
output sequence of \dsort is bounded above as follows.
\[
  n_{\mathit{max}} = \lf 1 + \frac{1}{\ceil{\omega_{n}}}\rg
    \frac{n}{p} + \ceil{\omega_{n}} p .
\]
\noindent
The lemma follows.
\end{proof}

\subsection{Duplicate-key Handling}
\label{duplicate}
\label{sd}

Algorithm \dsort, as described, does not handle
duplicate keys properly. A naive way
to handle duplicate keys is by making the keys distinct. This could
be achieved
by attaching to each key the address of the memory location it is
stored in.  For data types whose bit or byte size   is comparable to
the size   of the address describing them,
such a transformation leads -- in most cases -- to a doubling of the
overall number of comparisons performed and the communication time
in the worst case.  For
more complex data types such as strings of characters the  extra cost
may be negligible.

An alternative way to handle duplicate keys 
in a transparent way that provides asymptotic optimal efficiency and
tags only a small fraction of the keys
is the following  one.
This seems to be an improvement
over other approaches \cite{jaja1,jaja2,jaja3} that require
``doubling'' (as explained earlier.)
Procedure \sort must then be implemented
by means of a stable sequential sorting algorithm as well.
Two tags for each input key are already implicitly available by default,
and no extra memory is required to access them.
These are the sequence  identifier that stores a particular input key
and the index of the key in the local array that stores that sequence. 
No additional
space is required for the maintenance of this tagging.
In our duplicate-key handling method such tags are only used for sample
and splitter-related activity.

For sample sorting every sample key is augmented 
into a record that includes this additional tag information
(array index and sequence identifier storing the key).
Since sample size is $o(1)$ of the input keys, the memory overhead
incurred is small, as is the corresponding computational overhead.
The attached tag information is used in  step  2  to form the
sample and in sample sorting/merging and then in splitter selection,
and finally in step 4, as all
these steps require distinct keys to achieve stability.
In step 4 in particular, a binary search operation
of a splitter key into the locally sorted keys involves  first
a  comparison of the two keys, and if the comparison is not
resolved that way the use of sequence identifiers or array
indexes as well. If merging is used instead, a similar resolution
applies. In the multi-way merging of the {\sc Merging} phase,
stability is resolved by the merging algorithm itself.
The computation overhead of duplicate handling that
is described by this method
is within the lower order terms of the analysis
and therefore, the optimality claims still hold unchanged.
The results on key imbalance still hold as well.
This same duplicate handling method is also used in  \rsort.

\subsection{Operation \rsort}

In this section we show how to modify \dsort
to form randomized operation \rsort.
Random oversampling-based algorithms in the traditional
approach of \cite{HC,Reischuk85,Reif87,Gerbessiotis94} do not 
involve a {\sc BaselineSorting} phase that thus distinguishes
\dsort and \rsort from other similar approaches.

%
%
\begin{algorithm}
\caption{{\rsort $(X,n,p,$ \sort $)$  \{sort $n$ keys of $X$; utilize \sort for sorting \} }}
\label{RSORT}
\begin{algorithmic}[1]
\STATE {\sc BaselineSorting.} The $n$ input keys are regularly and evenly split into $p$ 
sequences each one of size   approximately $n/p$. Each sequence
is then  sorted by \sort. Let $X_k$, $0 \leq k \leq p-1$, be the $k$-th sequence after
sorting.
\STATE {\sc Partitioning: Sample Selection.}  Let $s =   2\omega_n^2 \lg{n} $. 
Form a sample $T$ from the corresponding $X_k$.  The sample consists of $sp-1$ keys
selected uniformly at random from the keys of all $X_k$. 
Sort $T$ using \sort.
\STATE {\sc Partitioning: Splitter Selection.}  
Form the splitter sequence $S$ that contains the $(i\cdot s)$-th smallest keys of $T$,
$1 \leq i \leq p-1$. 
\STATE {\sc Partitioning: Split input keys.}  
Split the sorted $X_k$ around $S$ into sorted subsequences $X_{k,j}$, $0 \leq j \leq p-1$, for all
$0 \leq k \leq p-1$.
\STATE  {\sc Merging.} All  sorted subsequences $X_{k,j}$ for all $0 \leq k \leq p-1$ are merged
into $Y_j$, for all $0 \leq j \leq p-1$. The concatenation of $Y_j$ for all $j$ is $Y$. 
Return $Y$.
\end{algorithmic}
\end{algorithm}

Although  partitioning and oversampling in the context of sorting
are well established techniques  \cite{HC,Reischuk85,Reif87}, 
the analysis in \cite{Gerbessiotis94} summarized in
Claim \ref{Sampling1} below allows one to 
quantify precisely the key imbalance of the output sequences $Y_j$.
Let $X = \langle x_{1}, x_{2}, \ldots , x_{n} \rangle$ be an ordered
sequence of keys indexed such that $x_{i} < x_{i+1}$, for all
$1 \leq i  \leq n-1$. The implicit assumption is that keys are unique.
Let
$Y = \{y_{1}, y_{2}, \ldots, y_{ps-1}\}$ be a randomly chosen 
subset of $ps-1 \leq n$ keys of $X$ also indexed such that 
$y_{i} < y_{i+1}$, for all $1 \leq i  \leq ps-2$,  for some 
positive integers $p$ and $s$. Having
randomly selected set $Y$, a partitioning of  $X - Y$ into $p$
subsets, $X_{0}, X_{1}, \ldots, X_{p-1}$ takes place.
The following result shown in \cite{Gerbessiotis94}
is independent of the distribution of the input keys.

\begin{cla}
\label{Sampling1}
Let $p \geq 2$, $s \geq 1$, $ps < n/2$, $n \geq 1$,  $0 < \varepsilon < 1$,
$\rho >0$, and

\[
  s \geq \frac{1+\varepsilon}{\varepsilon^{2}} \lf( 2 \rho \log{n} +
         \log{(2 \pi p^{2}(ps-1)e^{1/(3(ps-1))})} \rg).
\]
\noindent
Then the probability that any one of the $X_{i}$, for all $i$,
$0 \leq i \leq p-1$, is of size more than
$\ceil{(1+\varepsilon)(n-p+1)/p}$ is at most $n^{-\rho}$.
\end{cla}

In such traditional randomized algorithms that use oversampling the
complexity of splitting  keys around the splitters 
is inconsequential.  For each input key a binary search operation 
on the $p-1$ splitters determines the position of that 
input key in one of the $p$ output sequences that will
be formed and then sorted. This is much more involved than say 
step 4 of \dsort that can involve a binary search of the $p-1$ 
splitters into each one of the $p$ sorted sequences
of approximately $n/p$ keys, an operation that has better
locality of reference.

Operation \rsort based on \dsort is
described in Algorithm~\ref{RSORT}
The analysis of \rsort is identical to that of \dsort
described in Theorem~\ref{detsorting}. The major difference
involves sample sorting. The sample $ps$ can be sorted directly
in time $Aps \lg{(ps)}$ using \sort rather than merged  
in time $O(ps \lg{p})$.
However in both cases the asymptotics of these two terms
 remain identically the same $O(ps \lg{n})$.
Thus the running time contribution of \rsort can be
summarized as $A n\lg{n}+(A-B)n \lg{p} + O( ps \lg{n} + n)$.
From the latter, if $ps = o(n)$, then the running time
becomes
$A n\lg{n}+(A-B)n \lg{p} + o(n \lg{n})$. For $ps =o(n)$ we
need $2p \omega_n^2 \lg{n} = o(n)$.
Theorem \ref{iransorting} is then derived.

\begin{theo}
\label{iransorting}
For any $n$ and $p \leq n$, and any function
functions $\omega_n$ and $s=2 \omega_n^2 \lg{n}$ such that 
$p s <n/2$, and $ps = o(n)$,
algorithm \rsort is  such that its running time is
bounded by $A n\lg{n}+(A-B)n \lg{p} + o(n \lg{n})$,
if \sort requires $A n \lg{n}$ time to sort $n$ keys and 
it takes $B n \lg{p}$ time to merge $n$ keys of $p$ sorted sequences, 
for some constants $A,B \geq 1$.
\end{theo}


\subsection{Multi-core adaptations}

We perform a minimal and straightforward modification that 
involves minimal parallelization: the adaptation of step 1 and 
potentially step 5 for multi-core and other parallel 
architectures. In the experimental study of the
following section, only step 1 of \mrsort was thus modified to take
advantage of multi-core architectures.
If a processor with $m$ cores is available, the $p$
sorting operations of step 1 ({\sc BaselineSorting}) can be
distributed evenly among the $m$ cores. The contribution to the
runtime of this would be $(An/m) \lg{(n/p)}$. Similarly
for step 5 we would get a time reduction to $(Bn/m) \lg{p}$. For
all these claims to hold we must have that $m \leq p$ and 
$p$ should also be a multiple of $m$. 
Thus the runtime of 
Theorem~\ref{iransorting} or
Theorem~\ref{detsorting}
can be reduced by a multiplicative factor $m$
%
As we have already mentioned the multi-core implementation
of the experimental study forfeits the parallelization of
step 5. As a result its runtime is expected to behave
according to the expression $ (An/m) \lg{(n/p)} + Bn \lg{p}$,
where low-order terms are not shown. In fact it is written so that
it can spawn a number of threads $t$ which can be higher than the number
of available cores.

\section{Experimental Study}

Operations \dsort and \rsort developed under this work
have been implemented in ANSI C and their
performance studied for a variety of \sort functions with the
intent to verify whether the theoretical analysis can be verified
in practice.
Moreover, \mrsort, the version or \rsort that takes advantage
of multi-core architectures as described in the previous section was also
implemented.

The resulting source code is publically available through 
the author's web-page \cite{AVG16}. 
The implementation is in standard ANSI C.
A quad-core Intel Xeon E3-1240 3.3Ghz Scientific Linux 7 workstation with 
16GiB of memory has been used for the experiments.
The source code is  compiled using the native gcc compiler 
{\tt gcc version 4.8.5} with optimization options 
{\tt -O3 -mtune=native -march=native}
and using otherwise the default compiler and library installation.

Among the candidates for \sort used, function {\tt qsort} is the system
supplied ANCI C library function used for calibration and base reference; 
the other functions is an author-derived version of heapsort
and an author-derived version of recursive quicksort.
We call them {\tt qs, hs, rq} respectively.
Indicated timing results (wall-clock time in seconds)  in the tables to follow
are the averages of three experiments.
The input consists of random strings 32B long; 
string comparison is performed through the {\tt memcmp} 
ANSI C function.
Each byte takes values uniformly distributed between 0 and 255.

Timing data are reported for standalone
execution of  {\tt qs, hs, rq} when run independently of \dsort and \rsort
under a column labeled \sort.
Timing data for \dsort and \rsort are shown for each one of the three 
choices of \sort : {\tt qs, hs, rq},
for a variety of problem sizes such as
$n=1024000$, $n=4096000$,  $n=8192000$, $n=16384000$ and $n=32768000$,
and for various  splitter sizes $p-1$ such as  $p=4, 32, 64, 128, 256$; 
note that $p-1$ is the number of splitters used and $p$ is the number of split
subsequences induced by the $p-1$ splitters.
The default value of sample size is controlled by parameter $s$ in \rsort and 
it is  $s=2 \omega_n^2 \lg{n} =  \lg^2{(n)} $ rounded-up. 
Additional data are provided separately for the middle problem size of $n=8192000$: 
sample size other than the default value is being used then.
Thus in Table~\ref{Table2}, we vary $s$ by using parameter $a$ defined as 
$s=\log^{1+a}{n}$.  This is equivalent to choosing $\omega_n \approx \log^{a/2}{n}$.
For such a definition of $a$, values of  $a$ running between $0.2$ and $1.8$ in 
increments of $0.4$ are being used in the timing data reported.
For \dsort sample size is regulated by parameter  $r= \omega_n$ and we choose it
be a small integer in the range 1-5.

In Table~\ref{Table1}  timing results  for \rsort are reported, 
where \sort is each one of {\tt qs, hs, rq} and the default $s$, as specified
earlier, is used.
We observe that even for the smallest problem size, and the most time-efficient \sort 
(i.e. {\tt qsort}) operation \rsort offers better (though marginally)
performance than standalone {\tt qsort}.
For increasing problem sizes, the best savings occur for increasing values of $p$, 
up to a point that makes sample size which is dependent on $p$ a contributing factor.
One can potentially  extract better performance out of
\rsort by fine-tuning $p$ and $s$ to values other that the ones chosen. 
A smaller sample size that is a fraction of the theoretically chosen one can 
speed things up as well.
This is shown in Table~\ref{Table2} where for a fixed problem size of $n=8192000$ the
2.74 or so running time of {\tt qs} is bested not by the default $s$ (implied by $a=1.0$) 
but by a slight smaller $s$ obtained for $a=0.2$ or $a=0.4$. 
Such savings however are marginal and thus using the default value for $s$ is 
satisfactory enough.
For {\tt hs} or {\tt rq} the use of \rsort improves  upon the standalone use of
{\tt hs} or {\tt rq} by 30\% or more and 15\% or more respectively.

In Table~\ref{Table3} and Table~\ref{Table4} even the naive ''parallelization''
available through the modification of \rsort into \mrsort
provides substantial improvement in performance. The latter table does not
report timing data for $p=4$ as the number of threads $t$ is greater than $p=4$
and our code requires $p$ to be equal to at least the number of threads used.
The multi-core version of \rsort (i.e. \mrsort) that uses {\tt qsort} for \sort
and  four threads easily provides a speedup in performance of a factor of 1.6
(0.18 vs 0.30 sec for $n=1024000$) to 2.4 (4.96 vs 12.05 sec for $n=32768000$).
For {\tt hs} and {\tt rq} the reported speedup in performance in Table~\ref{Table3} 
is higher and can reach 4 or so for the former and 2.9 for the latter. 
Note that  using a quick-sort based algorithm ({\tt qs} or {\tt rq})  for \sort  in
a multi-threaded/multi-core set-up, best performance is obtained for $p=4$ i.e. for
a value equal to the number of cores.
This has to do with the inherent values  of $A$ and $B$ of
Theorem~\ref{iransorting} for {\tt qs, hs, rq}.
Note that the architecture used for the experimental study is a quad-core processor 
that supports up to 8 threads (two per core) in hardware. 
Thus we have been able to run \mrsort using 
8 threads so that we can take full advantage of the underlying hardware support.
As it is evident from Table~\ref{Table4} further improvement in performance is possible
and overall execution times are lower for all problem sizes and choices of the 
\sort function. Moreover Table~\ref{Table5} and Table~\ref{Table6} confirm our previous
findings: varying sample size can improve performance but only slightly relative to the
default value used.

Overall, all three choices of \sort can benefit from the use of our proposed operation \rsort.
The benefits are modest and marginal for {\tt qs} and maximized for the slowest of our
implementation. However when the multi-core version of our proposed operation
is employed i.e. \mrsort, that uses just a marginal and ''naive'' parallelization 
(eg assignment of {\sc BaselineSorting} tasks to different threadsc/cores), 
savings can be significant and speed up in performance
by a factor of 1.5 to  2.5 is possible even for the system's {\tt qsort} implementation.
Thus employing our proposed sorting operation \mrsort and using in it {\tt qsort} for \sort can lead to
two-fold to three-fold increase in performance relative to using the standalone {\tt qsort} function.


Finally Table~\ref{Table7} and Table~\ref{Table8} present some experimental results related to
operation \dsort. Table~\ref{Table7} is similar to Table~\ref{Table1}. Note that there is some
variablity in the results of column \sort between the two tables. With the exception of
possibly problem size $n=1024000$, the results of \rsort in Table~\ref{Table1} are slightly
better than those of \dsort consistently as they should be. 

For the small problem size, it is possible that the regularity of the sample
in \dsort affects overall performance marginally. 
The imbalance of the sizes of sequences $Y_j$ in step 5 of \dsort are indeed higher than those of \rsort
for the value of $r=1$ used in Table~\ref{Table7}  and the default $s$ of Table~\ref{Table1} .
Note that in the former case the imbalance as controlled by $\omega_n$ is $r=\lceil omega_n \rceil$
and $r=1$. For Table~\ref{Table1},  $\omega_n$ is much larger (approximately $\sqrt{\lg{n}}$) 
and thus the imbalance of $Y_j$ smaller for the default value of $s$.
Other than that for all the experiments the remarks related to Table~\ref{Table1} still apply.
And so do the remarks related to Table~\ref{Table2}. The role of $r$ in controlling 
deterministic sample size in \dsort is assigned to $a$ in \rsort.
Thus we saw no reason to modify \dsort the way we modified \rsort to generate \mrsort.

\section{Conclusion}

We have presented  sequential sorting operations inspired by parallel
computing techniques and developed new sequential sorting methods that can 
improve the performance of generic and optimized sorting algorithm 
implementations available in various programming libraries or provided by
programmers. 
We have  implemented one deterministic and one randomized sorting
operation using this method for   a variety of auxiliary generic
\sort functions including the C Standard library available {\tt qsort} 
and  studied and compared  the performance of
standalone \sort operations against our  proposed
\dsort and \rsort operations. Our \dsort and \rsort operations improved
the performance of standalone \sort in all cases.
In addition we have
presented a simple and easy to develop multi-core implementation of 
the \rsort method denoted \mrsort. 
This multi-core implementation, even with the
current limitations of our implementation,
shows significant, though expected, performance improvements against 
optimized sorting implementations such as the system available {\tt qsort}
function.
The conclusion drawn as
a result of the experimental study that we undertook
is that parallel computing techniques
designed to effect blocked interprocessor communication and to 
take advantage of locality of reference can provably benefit sequential 
computing as well, and can lead to a new set of sorting algorithms. 
Some code used in the experimental study reported in this work was from a prior
project supported in part by NSF grant NSF/ITR ISS-0324816.





\newpage

\begin{table}[h!b!p!]\centering
{\begin{tabular}{|l r|r|r|r|r|r|}\hline
\multicolumn{2}{|c|}{} &
\multicolumn{5}{c|}{\rsort}       \\ \hline
\multicolumn{2}{|c|}{Problem size : \sort (sec)} &
\multicolumn{1}{c|}{$p=4$}         &
\multicolumn{1}{c|}{$p=32$}        &
\multicolumn{1}{c|}{$p=64$}        &
\multicolumn{1}{c|}{$p=128$}       &
\multicolumn{1}{c|}{$p=256$}       \\ \hline
$n=1024000$  : qs & 0.30  & 0.29 & 0.28 & 0.28 & 0.28 & 0.30 \\
$n=1024000$  : hs & 0.56  & 0.46 & 0.41 & 0.41 & 0.40 & 0.42 \\
$n=1024000$  : rq & 0.40  & 0.37 & 0.35 & 0.35 & 0.35 & 0.36 \\ \hline
$n=4096000$  : qs & 1.32  & 1.34 & 1.28 & 1.25 & 1.25 & 1.26 \\
$n=4096000$  : hs & 2.73  & 2.35 & 1.87 & 1.83 & 1.81 & 1.77 \\
$n=4096000$  : rq & 1.74  & 1.67 & 1.56 & 1.55 & 1.53 & 1.53 \\ \hline
$n=8192000$  : qs & 2.74  & 2.78 & 2.71 & 2.66 & 2.64 & 2.61 \\
$n=8192000$  : hs & 6.06  & 5.25 & 4.03 & 3.84 & 3.79 & 3.67 \\
$n=8192000$  : rq & 3.64  & 3.42 & 3.30 & 3.25 & 3.21 & 3.17 \\ \hline
$n=16384000$ : qs & 5.78  & 5.86 & 5.71 & 5.64 & 5.54 & 5.46 \\
$n=16384000$ : hs &13.12  &11.64 & 9.08 & 8.29 & 7.90 & 7.75 \\
$n=16384000$ : rq & 7.32  & 7.36 & 7.06 & 6.82 & 6.71 & 6.64 \\ \hline
$n=32768000$ : qs &12.05  &12.22 &12.05 &11.86 &11.70 &11.48 \\
$n=32768000$ : hs &28.89  &25.40 &20.37 &18.64 &17.16 &16.39 \\
$n=32768000$ : rq &15.13  &15.10 &14.74 &14.64 &14.33 &13.86 \\ \hline
\end{tabular}}
\caption{\rsort : varying problem size $n$ and splitter size $p-1$
\label{Table1}}
\end{table}

\begin{table}[h!b!p!]\centering
{\begin{tabular}{|l r|r|r|r|r|r|r|r|r|r|}\hline
\multicolumn{2}{|c|}{Splitter size ($p-1$) : \sort (sec)} &
\multicolumn{1}{c|}{$a=0.2$}       &
\multicolumn{1}{c|}{$a=0.6$}       &
\multicolumn{1}{c|}{$a=1.0$}       &
\multicolumn{1}{c|}{$a=1.4$}       &
\multicolumn{1}{c|}{$a=1.8$}       \\ \hline
$p=64$  : qs & 2.75   & 2.65 & 2.66 & 2.66 & 2.66 & 2.71 \\
$p=64$  : hs & 6.07   & 3.88 & 3.87 & 3.88 & 3.91 & 3.95 \\
$p=64$  : rq & 3.64   & 3.24 & 3.24 & 3.24 & 3.27 & 3.30 \\ \hline
$p=256$ : qs & 2.75   & 2.59 & 2.61 & 2.62 & 2.66 & 2.80 \\
$p=256$ : hs & 6.07   & 3.68 & 3.65 & 3.67 & 3.73 & 4.05 \\
$p=256$ : rq & 3.64   & 3.15 & 3.16 & 3.18 & 3.23 & 3.48 \\ \hline
\end{tabular}}
\caption{\rsort : varying sample parameters $a$ and $p$  for $n=8192000$
\label{Table2}}
\end{table}


\begin{table}[h!b!p!]\centering
{\begin{tabular}{|l r|r|r|r|r|r|}\hline
\multicolumn{2}{|c|}{} &
\multicolumn{5}{c|}{\mrsort}       \\ \hline
\multicolumn{2}{|c|}{Problem size : \sort (sec)} &
\multicolumn{1}{c|}{$p=4$}         &
\multicolumn{1}{c|}{$p=32$}        &
\multicolumn{1}{c|}{$p=64$}        &
\multicolumn{1}{c|}{$p=128$}       &
\multicolumn{1}{c|}{$p=256$}       \\ \hline
$n=1024000$  : qs & 0.30  & 0.18 & 0.21 & 0.20 & 0.21 & 0.25 \\
$n=1024000$  : hs & 0.56  & 0.25 & 0.25 & 0.27 & 0.27 & 0.29 \\
$n=1024000$  : rq & 0.40  & 0.19 & 0.23 & 0.24 & 0.25 & 0.28 \\ \hline
$n=4096000$  : qs & 1.32  & 0.62 & 0.66 & 0.69 & 0.69 & 0.81 \\
$n=4096000$  : hs & 2.73  & 0.93 & 0.86 & 0.87 & 0.90 & 0.96 \\
$n=4096000$  : rq & 1.74  & 0.65 & 0.74 & 0.76 & 0.80 & 0.83 \\ \hline
$n=8192000$  : qs & 2.74  & 1.20 & 1.32 & 1.38 & 1.43 & 1.45 \\
$n=8192000$  : hs & 6.06  & 1.90 & 1.78 & 1.73 & 1.71 & 2.11 \\
$n=8192000$  : rq & 3.64  & 1.28 & 1.53 & 1.48 & 1.69 & 1.63 \\ \hline
$n=16384000$ : qs & 5.78  & 2.43 & 2.73 & 2.81 & 2.85 & 2.86 \\
$n=16384000$ : hs &13.12  & 4.00 & 3.81 & 3.72 & 3.59 & 3.55 \\
$n=16384000$ : rq & 7.32  & 2.54 & 2.91 & 3.01 & 3.03 & 3.23 \\ \hline
$n=32768000$ : qs &12.05  & 4.96 & 5.71 & 5.90 & 6.08 & 6.04 \\
$n=32768000$ : hs &31.74  & 8.66 & 8.22 & 8.13 & 7.78 & 7.53 \\
$n=32768000$ : rq &15.13  & 5.19 & 6.01 & 6.21 & 6.89 & 6.53 \\ \hline
\end{tabular}}
\caption{\mrsort : 4 threads, varying problem size $n$ and splitter size $p-1$
\label{Table3}}
\end{table}



\begin{table}[h!b!p!]\centering
{\begin{tabular}{|l r|r|r|r|r|r|}\hline
\multicolumn{2}{|c|}{} &
\multicolumn{5}{c|}{\mrsort}       \\ \hline
\multicolumn{2}{|c|}{Problem size : \sort (sec)} &
\multicolumn{1}{c|}{$p=4$}         &
\multicolumn{1}{c|}{$p=32$}        &
\multicolumn{1}{c|}{$p=64$}        &
\multicolumn{1}{c|}{$p=128$}       &
\multicolumn{1}{c|}{$p=256$}       \\ \hline
$n=1024000$  : qs & 0.30  &      & 0.12 & 0.13 & 0.16 & 0.17 \\
$n=1024000$  : hs & 0.56  &      & 0.15 & 0.16 & 0.18 & 0.20 \\
$n=1024000$  : rq & 0.40  &      & 0.13 & 0.14 & 0.15 & 0.18 \\ \hline
$n=4096000$  : qs & 1.32  &      & 0.52 & 0.55 & 0.64 & 0.64 \\
$n=4096000$  : hs & 2.73  &      & 0.67 & 0.68 & 0.70 & 0.76 \\
$n=4096000$  : rq & 1.74  &      & 0.54 & 0.59 & 0.62 & 0.68 \\ \hline
$n=8192000$  : qs & 2.74  &      & 1.12 & 1.13 & 1.19 & 1.28 \\
$n=8192000$  : hs & 6.06  &      & 1.47 & 1.46 & 1.47 & 1.52 \\
$n=8192000$  : rq & 3.64  &      & 1.13 & 1.19 & 1.27 & 1.36 \\ \hline
$n=16384000$ : qs & 5.78  &      & 2.35 & 2.42 & 2.46 & 2.60 \\
$n=16384000$ : hs &13.12  &      & 3.18 & 3.16 & 3.13 & 3.17 \\
$n=16384000$ : rq & 7.32  &      & 2.40 & 2.49 & 2.59 & 2.76 \\ \hline
$n=32768000$ : qs &12.05  &      & 4.99 & 5.08 & 5.24 & 5.38 \\
$n=32768000$ : hs &28.89  &      & 6.88 & 6.82 & 6.77 & 6.71 \\
$n=32768000$ : rq &15.13  &      & 5.01 & 5.23 & 5.51 & 5.64 \\ \hline
\end{tabular}}
\caption{\mrsort : 8 threads, varying problem size $n$ and splitter size $p-1$
\label{Table4}}
\end{table}

\begin{table}[h!b!p!]\centering
{\begin{tabular}{|l r|r|r|r|r|r|r|r|r|r|}\hline
\multicolumn{2}{|c|}{Splitter size ($p-1$) : \sort (sec)} &
\multicolumn{1}{c|}{$a=0.2$}       &
\multicolumn{1}{c|}{$a=0.6$}       &
\multicolumn{1}{c|}{$a=1.0$}       &
\multicolumn{1}{c|}{$a=1.4$}       &
\multicolumn{1}{c|}{$a=1.8$}       \\ \hline
$p=64$  : qs & 2.75   & 1.36 & 1.36 & 1.38 & 1.38 & 1.40 \\
$p=64$  : hs & 6.07   & 1.75 & 1.77 & 1.75 & 1.78 & 1.80 \\
$p=64$  : rq & 3.64   & 1.52 & 1.51 & 1.50 & 1.51 & 1.55 \\ \hline
$p=256$ : qs & 2.75   & 1.48 & 1.43 & 1.50 & 1.55 & 1.64 \\
$p=256$ : hs & 6.07   & 1.80 & 1.82 & 1.78 & 1.85 & 2.11 \\
$p=256$ : rq & 3.64   & 1.62 & 1.64 & 1.65 & 1.69 & 1.88 \\ \hline
\end{tabular}}
\caption{\mrsort : 4threads, varying sample parameters $a$ and $p$  for $n=8192000$
\label{Table5}}
\end{table}

\begin{table}[h!b!p!]\centering
{\begin{tabular}{|l r|r|r|r|r|r|r|r|r|r|}\hline
\multicolumn{2}{|c|}{Splitter size ($p-1$) : \sort (sec)} &
\multicolumn{1}{c|}{$a=0.2$}       &
\multicolumn{1}{c|}{$a=0.6$}       &
\multicolumn{1}{c|}{$a=1.0$}       &
\multicolumn{1}{c|}{$a=1.4$}       &
\multicolumn{1}{c|}{$a=1.8$}       \\ \hline
$p=64$  : qs & 2.75   & 1.13 & 1.13 & 1.14 & 1.15 & 1.17 \\
$p=64$  : hs & 6.07   & 1.45 & 1.45 & 1.46 & 1.47 & 1.52 \\
$p=64$  : rq & 3.64   & 1.19 & 1.20 & 1.19 & 1.20 & 1.24 \\ \hline
$p=256$ : qs & 2.75   & 1.26 & 1.26 & 1.27 & 1.34 & 1.47 \\
$p=256$ : hs & 6.07   & 1.50 & 1.50 & 1.53 & 1.76 & 1.84 \\
$p=256$ : rq & 3.64   & 1.34 & 1.34 & 1.37 & 1.43 & 1.67 \\ \hline
\end{tabular}}
\caption{\mrsort : 8threads, varying sample parameters $a$ and $p$  for $n=8192000$
\label{Table6}}
\end{table}

\begin{table}[h!b!p!]\centering
{\begin{tabular}{|l r|r|r|r|r|r|}\hline
\multicolumn{2}{|c|}{} &
\multicolumn{5}{c|}{\dsort}       \\ \hline
\multicolumn{2}{|c|}{Problem size : \sort (sec)} &
\multicolumn{1}{c|}{$p=4$}         &
\multicolumn{1}{c|}{$p=32$}        &
\multicolumn{1}{c|}{$p=64$}        &
\multicolumn{1}{c|}{$p=128$}       &
\multicolumn{1}{c|}{$p=256$}       \\ \hline
$n=1024000$  : qs & 0.30  & 0.30 & 0.29 & 0.26 & 0.29 & 0.31 \\
$n=1024000$  : hs & 0.55  & 0.46 & 0.42 & 0.38 & 0.41 & 0.42 \\
$n=1024000$  : rq & 0.37  & 0.36 & 0.34 & 0.34 & 0.34 & 0.35 \\ \hline
$n=4096000$  : qs & 1.31  & 1.34 & 1.31 & 1.31 & 1.29 & 1.30 \\
$n=4096000$  : hs & 2.78  & 2.39 & 1.92 & 1.89 & 1.84 & 1.81 \\
$n=4096000$  : rq & 1.66  & 1.61 & 1.54 & 1.54 & 1.52 & 1.53 \\ \hline
$n=8192000$  : qs & 2.74  & 2.81 & 3.07 & 3.01 & 2.72 & 2.71 \\
$n=8192000$  : hs & 6.11  & 5.30 & 4.10 & 4.00 & 3.88 & 3.80 \\
$n=8192000$  : rq & 3.48  & 3.31 & 3.27 & 3.22 & 3.20 & 3.17 \\ \hline
$n=16384000$ : qs & 5.75  & 5.87 & 5.92 & 5.78 & 5.73 & 6.25 \\
$n=16384000$ : hs &13.25  &11.65 & 9.21 & 8.56 & 8.24 & 8.08 \\
$n=16384000$ : rq & 6.81  & 6.99 & 6.94 & 6.73 & 6.68 & 6.52 \\ \hline
$n=32768000$ : qs &12.04  &12.27 &12.27 &12.23 &12.14 &11.97 \\
$n=32768000$ : hs &29.14  &26.11 &20.58 &18.99 &17.49 &17.00 \\
$n=32768000$ : rq &14.41  &14.80 &14.45 &14.34 &14.07 &13.90 \\ \hline
\end{tabular}}
\caption{\dsort for $r=1$ : varying problem size $n$ and splitter size $p-1$
\label{Table7}}
\end{table}

\begin{table}[h!b!p!]\centering
{\begin{tabular}{|l r|r|r|r|r|r|r|r|r|r|}\hline
\multicolumn{2}{|c|}{} &
\multicolumn{5}{c|}{\dsort}       \\ \hline
\multicolumn{2}{|c|}{Splitter size ($p-1$) : \sort (sec)} &
\multicolumn{1}{c|}{$r=1  $}       &
\multicolumn{1}{c|}{$r=2  $}       &
\multicolumn{1}{c|}{$r=3  $}       &
\multicolumn{1}{c|}{$r=4  $}       &
\multicolumn{1}{c|}{$r=5  $}       \\ \hline
$p=64$  : qs & 2.74   & 2.74 & 2.74 & 2.73 & 2.73 & 2.74 \\
$p=64$  : hs & 6.11   & 3.99 & 4.00 & 3.99 & 4.00 & 4.00 \\
$p=64$  : rq & 3.48   & 3.22 & 3.22 & 3.22 & 3.22 & 3.23 \\ \hline
$p=128$ : qs & 2.74   & 2.72 & 2.72 & 2.72 & 2.72 & 2.73 \\
$p=128$ : hs & 6.11   & 3.92 & 3.88 & 3.88 & 3.88 & 3.89 \\
$p=128$ : rq & 3.48   & 3.20 & 3.20 & 3.20 & 3.20 & 3.27 \\ \hline
$p=256$ : qs & 2.74   & 2.74 & 2.72 & 2.74 & 2.74 & 2.75 \\
$p=256$ : hs & 6.11   & 3.84 & 3.81 & 3.83 & 3.84 & 3.85 \\
$p=256$ : rq & 3.48   & 3.20 & 3.18 & 3.19 & 3.20 & 3.21 \\ \hline
\end{tabular}}
\caption{\dsort :  varying sample parameters $r$ and $p$  for $n=8192000$
\label{Table8}}
\end{table}
\end{document}